\documentclass[12pt]{article}
\usepackage{amsmath,amssymb,amsfonts,amsthm,graphicx}
\usepackage[bookmarksnumbered, colorlinks, plainpages]{hyperref}
\usepackage{amsmath,amsthm,amscd,amsfonts,amssymb,enumerate}
\usepackage{graphicx}		
\usepackage{color}
\usepackage[colorlinks]{hyperref}
 
\newtheorem{theorem}{Theorem}[section]
\newtheorem{proposition}[theorem]{Proposition}
\newtheorem{lemma}[theorem]{Lemma}
\newtheorem{corollary}[theorem]{Corollary}
\theoremstyle{definition}
\newtheorem{definition}[theorem]{Definition}
\newtheorem{example}[theorem]{Example}
\theoremstyle{remark}
\newtheorem{remark}[theorem]{Remark}
\numberwithin{equation}{section}
 
 \newcommand{\ud}{\,\mathrm{d}}
\def\A{\mathfrak{A}}
\def\XA{\mathfrak{X}(M)^{\mathfrak{A}}}
\def\XM{\mathfrak{X}M}
\def\TX{\tilde{X}}
\def\TY{\tilde{Y}}
\def\TZ{\tilde{Z}}
\def\R{\mathbb{R}}
\title{Operator-Valued Tensors on Manifolds: \\A Framework for Field Quantization}
\author{\large H. Feizabadi$^1$, \hspace{.0cm}\large N. Boroojerdian$^2$\footnote{Corresponding author. $\newline$
{\em E-mail addresses:}broojerd@aut.ac.ir(N. Boroojerdian), hassan\_feiz1970@aut.ac.ir(H. Feizabadi). $\newline$
} \hspace{.0cm} \vspace{.5cm}$ $\\
\small{   $^{1,2}$Department of pure Mathematics,
Faculty of Mathematics and Computer Science,}\vspace{-1mm}\\
\small{ Amirkabir University of Technology, No. 424, Hafez Ave.,
Tehran, Iran.}}
\date{}

\begin{document}
 
\maketitle \thispagestyle{empty}\

\begin{abstract}
In this paper we try to prepare a framework for field quantization. To this end,  we aim to replace the field of scalars $ \mathbb{R} $ by self-adjoint elements of a commutative $ C^\star $-algebra, and reach an appropriate generalization of  geometrical concepts on manifolds. First, we put forward the concept of operator-valued tensors and extend semi-Riemannian metrics to operator valued metrics. Then, in this new geometry, some essential concepts of Riemannian geometry such as curvature tensor, Levi-Civita connection, Hodge star operator, exterior derivative, divergence,... will be considered.
\end{abstract}
\setcounter{section}{-1}
\section{\bf Introduction}
\label{intro}
 The aim of the present paper is to extend the theory of semi-Riemannian metrics to operator-valued semi-Riemannian metrics, which may be used as a framework for field quantization. Spaces of linear operators are directly related to $C^*$-algebras, so $C^*$-algebras are considered as a basic concept in this article. This paper has its roots in Hilbert $C^\star$-modules, which are frequently used in the theory of operator algebras, allowing one to obtain information about $C^\star$-algebras by studying Hilbert $C^\star$-modules over them. Hilbert $C^\star$-modules provide a natural generalization of Hilbert spaces arising when the field of scalars $\mathbb{C}$ is replaced by an arbitrary $C^\star$-algebra. This generalization, in the case of commutative $C^\star$-algebras appeared in the paper of Kaplansky \cite{kap}, however the non-commutative case seemed too complicated at that time.The general theory of Hilbert $C^\star$-modules appeared in the pioneering papers of W. Paschke \cite{Pasch} and M. Rieffel \cite{rifl}.The theory of Hilbert $C^\star$-modules may also be considered as a non-commutative generalization of the theory of vector bundles and non-commutative geometry \cite{cun}.

A number of results about geometrical structures of Hilbert $C^\star$-modules and about operators on them have been obtained \cite{fran}.
Henceforth, Hilbert $C^\star$-modules are generalization of inner product spaces that on the level of manifolds, provide a generalization of semi-Riemannian manifolds,that is the goal of this paper. Due to the physical applications in mind, in this paper, the positive definiteness of the inner product will be replaced by non-degeneracy. Another root of this paper is to provide a framework for field quantization. The main idea for quantization is to replace scalars by operators on a Hilbert space. In the field of quantum mechanics, spectrum of operators plays the role of values of the measurements. So replacements of scalars by operators is the first step for quantization. In this direction, it has been done many works \cite{avr1}\cite{avr2}, but it seems $C^*$-algebras are the best candid to play the role of scalars and we must deal with Modules over $C^*$-algebras.  Modules, discussed in this article are free finite dimensional modules, and have bases that simplify many computations.

In this article, we only consider commutative $C^*$-algebras for many reason. Non-commutative algebras can be used in the realm of non-commutative Geometry, and it is not our aim to enter in this realm. Many basic definition such as vector field can not be extended properly for non-commutative $C^*$-algebras, because the set of derivations of an algebra is a module on the center of that algebra. So, extension of vector fields are modules on the center of that $C^*$-algebra and we need center of algebra be equal to the algebra, so the algebra must be commutative. Definition of the inner product encounter the same problem. Bilinear maps over a module whose scalars are non-commutative are very restricted. From Physical point of view, commutativity means operators in the $C^*$-algebra represent quantities that are simultaneously measurable and this assumption is not very restricted.

The content of the present paper is structured as follows: Section $1$ contains the preliminary facts about of $C^{\star}$-algebras needed to explain our concepts, Section $2$ covers the definition of extended tangent bundle, operator-valued vector fields, and operator-valued tensors and explains some of their basic properties. The Pettis-integral of operator-valued volume forms are defined and Stokes's theorem is proved in section $3$. Section $4$ is devoted to operator-valued connections and curvature, and the covariant derivative of operator-valued vector fields will be defined.  In section $5$, operator-valued inner product and some of their basic properties will be illustrated. The existence,uniqueness and the properties of the Hodge star operator for operator-valued inner product spaces are the goal of Section 6. In section $7$ the concepts of section $5$ extend to manifolds. The existence and uniqueness of Levi-Civita connection of operator-valued semi-Riemannian metrics is proved in section $8$. In the last section Ricci tensor, scalar curvature, and sectional curvature are deliberated.
\section{\bf Review of $C^{\star}$-Algebras}
In this section we review some definitions and results from $C^{\star}$-algebras that we need in the sequel.
\begin{definition}
A Banach $\star$-algebra is a complex Banach algebra $\mathfrak{A}$ with a conjugate linear involution $\ast$ which is an anti-isomorphism. That is, for all $A,B\in\mathfrak{A}$ and $\lambda\in\mathbb{C}$,
$$\left( A+B\right) ^{\ast }=A^{\ast }+B^{\ast },\quad 
\left( \lambda A\right) ^{\ast } =\overline{\lambda }A^{\ast },\quad
A^{\ast \ast } =A,\quad
\left( AB\right) ^{\ast } =B^{\ast }A^{\ast }$$.

\end{definition}
\begin{definition}
A $C^{\star}$-algebra $\A$, is a Banach $\star$-algebra with the additional norm condition. For all $ A\in\mathfrak{A}$
\begin{eqnarray*}
\left\Vert A^{\ast}A\right\Vert =\left\Vert A\right\Vert ^{2}.
\end{eqnarray*}
\end{definition}

For example the space of all bounded linear operators on a Hilbert space of $\mathcal{H}$ is a $C^{\star}$-algebra. This $C^{\star}$-algebra is denoted by $\mathfrak{B}\left(\mathcal{H}\right)$
\begin{remark}
We consider only unital $C^{\star}$-algebras, and it's unit element is denoted  by $1$.
\end{remark}

By Gelfand-Naimark theorem, all unital commutative $C^{\star}$-algebras have the form $C(X)$, in which $X$ is a compact Hausdorff space.
\begin{definition}
 An element $A$ of a $C^{\star}$-algebra is said to be
 \textbf{self-adjoint} if $A^{\ast}=A$, \textbf{normal} if $A^{\ast }A= AA^{\ast }$, and \textbf{unitary} if $A^{\ast }A= AA^{\ast }=1$.  
\end{definition}

 For now on, $\A$ is a $C^*$-algebra. The set of all self-adjoint elements of $\mathfrak{A}$ is denoted by $\mathfrak{A}_{\mathbb{R}}$.
\begin{definition}
The spectrum of an element $A$ in a  $C^{\star}$-algebra is the set
\[\sigma(A)=\{z\in\mathbb{C}: z1-A\: \mbox{is not invertible}\}.\]
\end{definition}
\begin{theorem}
The spectrum $\sigma(A)$ of any element A of a $C^{\star}$-algebra is a non empty compact set and
contained in the $\{z\in\mathbb{C}
:\left\vert z\right\vert \leqslant\left\Vert A\right\Vert \}$,
 if $A\in\mathfrak{A}_{\mathbb{R}}$, then $\sigma(A)\subseteq\mathbb{R}.$  (\cite{kan},\cite{rud})
 \end{theorem}
 
\begin{remark}
For any $A\in\A$, if $\lambda\in\sigma(A)$, then $\vert\lambda\vert\leq\Vert A\Vert$.
\end{remark}
For each normal element $A\in\mathfrak{A}$ there is a smallest $C^{\star}$-subalgebra $C^{\ast}(A,1)$ of $\mathfrak{A}$ which contains $A$, $1$, and is isomorphic to $C(\sigma(A))$.

 An element $A\in\mathfrak{A}_{\mathbb{R}}$ is said positive if $\sigma(A)\subseteq\mathbb{R}_+$. The set of positive elements of $\mathfrak{A}_{\mathbb{R}}$ is denoted by $\mathfrak{A}_{\mathbb{R}_+}$. For any $A\in\mathfrak{A}$,  $A^{\ast}A$ is positive.
\begin{theorem} 
If $A\in\mathfrak{A}_{\mathbb{R}_+}$, then there exists a unique element $B\in\mathfrak{A}_{\mathbb{R}_+}$ such that $B^2=A$. \cite{mu}
\end{theorem}

We denote by $\sqrt{A}$ the unique positive element $B$ such that $B^2=A$.
If $A$ is a self-adjoint element, then $A^2$ is positive, and we set $|A|=\sqrt{A^2}$,  $A^+=\frac{1}{2}(|A|+A)$, $A^-=\frac{1}{2}(|A|-A)$. Elements
 $|A|$, $A^+$, and $A^-$ are positive and $A=A^+-A^-$, $A^+A^-=0$. If $A,B\in\mathfrak{A}_{\mathbb{R}}$, then $\vert AB\vert=\vert A\vert \vert B\vert$. (cf. \cite{mu}).
\section{\bf Extending Tangent Bundle}
Throughout this paper, $\mathfrak{A}$ is a commutative unital $C^{\star}$-algebra which according to the \textbf{Gelfand-Naimark} second theorem can be thought as a $C^{\star}$-subalgebra of some $\mathfrak{B}\left(\mathcal{H}\right)$.
Let $M$ be a smooth manifold, we set
$TM^{\mathfrak{A}}=\underset{p\in M}{\cup}(T_{p}M\otimes_{\mathbb{R}}\mathfrak{A)}$, so $TM^{\mathfrak{A}}$ is a bundle of free $\mathfrak{A}$-modules over $M$.  Smooth functions from $M$ to $\mathfrak{A}$ can be defined and the set of them is denoted by $C^{\infty}(M,\mathfrak{A}).$ 
Addition, scalar multiplication, and  multiplication of functions in $C^{\infty}(M,\mathfrak{A})$ are defined pointwise. The involution of $\mathfrak{A}$ can be extended to $C^{\infty}(M,\mathfrak{A})$ as follows:
\begin{align*}
\ast:C^\infty(M,\mathfrak{A}&)\longrightarrow C^\infty(M,\mathfrak{A})\\
f&~\longmapsto f^*
\end{align*}
where $ f^{\ast}(x)=f(x)^{\ast} $.

Let $\A=C(X)$ for some compact Hausdorff space $X$. So, any function  $\tilde{f}:M\longrightarrow\A$ corresponds to a function $f:M\times X\longrightarrow\mathbb{C} $. Therefore, when $\A$ come to scene it means that we transfer from $M$ to $M\times X$ as the base manifold. In Physical applications, $M$ is a space-time manifold and $X$ may have variety of interpretations. For example, $X$ may be viewed as internal space of particles that produce quantum effects. In general, $X$ has no intrinsic relation to $M$.
\begin{definition}
A $\mathfrak{A}$-vector field $\tilde{X}$ over $M$ is a section of the bundle $ TM^{\mathfrak{A}} $.
\end{definition}

The set of all smooth $\mathfrak{A}$-vector fields on $M$ is denoted by $\mathfrak{X}(M)^{\mathfrak{A}}$, in fact
\[
 \mathfrak{X}(M)^{\mathfrak{A}}=\mathfrak{X}(M)\otimes_{C^{\infty}(M)} C^{\infty}(M,\mathfrak{A}). 
 \]
 Smooth $\mathfrak{A}$-vector fields can be multiplied by smooth
$\mathfrak{A}$-valued functions and $\mathfrak{X}(M)^{\mathfrak{A}}$ is a module over the $\star$-algebra $C^{\infty}(M,\mathfrak{A})$.
For a vector field $X\in\mathfrak{X}(M)$ and a function $ f\in C^{\infty}(M,\mathfrak{A}) $, we define $X\otimes f$ as a $\A$-vector field by
\[
(X\otimes f)_{p}=X_{p}\otimes f(p).
\]
These fields are called  simple, and any $\mathfrak{A}$-vector field can be written locally as a finite sum of simple $\mathfrak{A}$-vector fields. We can identify $X$ and $X\otimes 1$, so $\mathfrak{X}(M)$ is a $C^{\infty}(M)$-subspace of  $\mathfrak{X}(M)^{\mathfrak{A}}$.

A smooth vector field $X\in\mathfrak{X}(M)$ defines a derivation on $C^{\infty}(M,\mathfrak{A})$ by $ f\longmapsto Xf. $ In fact, for any integral curve $ \alpha:I \longrightarrow M $ of $X$, we have, $(Xf)(\alpha(t))=\frac{\ud}{\ud t}f(\alpha(t))$. 

This operation can be extended to $\mathfrak{A}$-vector fields, such that for simple elements of 
$\mathfrak{X}(M)^{\mathfrak{A}}$ such as $X\otimes h$ we have
\[
(X\otimes h)(f)=h.(Xf)\qquad f\in C^{\infty}(M,\mathfrak{A}) 
\]
If $\A$ is non-commutative, this definition is not well-defined. This definition implies that for $f, h\in C^{\infty}(M,\mathfrak{A})$ and $\widetilde{X}\in\mathfrak{X}(M)^{\mathfrak{A}}$ we have 
\[ (h\widetilde{X})(f)=h.(\widetilde{X}f). \]
The Lie bracket of ordinary vector fields can be extended to $\mathfrak{A}$-vector fields as the following, if $X,Y\in\mathfrak{X}(M)$ and $h,k\in C^{\infty}(M,\mathfrak{A})$ then,
\[
\left[  X\otimes h\ ,Y\otimes k\right]  =\left[  X ,Y\right]  \otimes (hk)+Y\otimes(hX(k))-X\otimes(kY(h)).
\]
The verification of main properties of the Lie bracket are routine.\\
An involution on $\mathfrak{X}(M)^{\mathfrak{A}}$ for simple $\mathfrak{A}$-vector field, such as $  X\otimes f $, is defined by 
\[
(X\otimes f)^{\ast}=X\otimes f^{\ast}. 
\]
For  $ \widetilde{X}, \widetilde{Y}\in \mathfrak{X}(M)^{\mathfrak{A}}, f\in  C^{\infty}(M,\mathfrak{A}) $, we have 
\begin{equation}
\widetilde{X}(f)^\ast=\widetilde{X}^\ast(f^\ast), 
 \left[\widetilde{X},\widetilde{Y}\right]^\ast= \left[\widetilde{X}^\ast,\widetilde{Y}^\ast\right].
 \end{equation} 
 \begin{definition}
A $ \mathfrak{A} $-valued covariant tensor field of order $k$ on $M$ is an operator $T:\XA\times\cdots\times\XA\longrightarrow C^{\infty}(M,\mathfrak{A})$ such that is $k-C^{\infty}(M,\A)$-linear.
\end{definition}
Contravariant and mixed $\mathfrak{A}$-valued tensors can be defined in a similar way.
Alternating covariant $\A$-tensor fields are called $\mathfrak{A}$-differential forms, and the set of all $\A$-differential forms of order $k$ is denoted by $A^k(M,\A)$. $\mathfrak{A}$-differential forms and exterior product and exterior derivation of these forms are special case of vector valued differential forms. The only difference is that $\mathfrak{A}$ maybe infinite dimensional.

Lie derivation of $\A$-tensor fields along $\A$-vector fields is defined naturally. For $\tilde{X}\in\XA$ and a covariant $\A$-tensor field of order $k$, such as $\tilde{T}$, Lie derivation of $\tilde{T}$ along $\tilde{X}$ is also a covariant $\A$-tensor field of order $k$ and defined as follows. For $\tilde{Y_1},\cdots ,\tilde{Y_k}\in\XA$ we have
\[
(L_{\tilde{X}} \tilde{T})(\tilde{Y_1},\cdots ,\tilde{Y_k})=\tilde{X}(\tilde{T}(\tilde{Y_1},\cdots ,\tilde{Y_k}))-\sum_{i=1}^k \tilde{T}(\tilde{Y_1},\cdots ,[\tilde{X},\tilde{Y_i}],\cdots ,\tilde{Y_k})
\]
If $\tilde{T}$ is a differential form, then its Lie derivation along any $\A$-vector field is also a differential form and Cartan formula for its derivation holds 
.i.e. 
\[
L_{\tilde{X}} \tilde{T}=d(i_{\tilde{X}}\tilde{T})+i_{\tilde{X}}(d\tilde{T})
\]

\section{\bf Integration of $\mathfrak{A}$-Valued Volume Forms and Stokes's Theorem}
 We remind the integral of vector valued functions, called \textbf{Pettis}-integral \cite{Uhl}. The Borel $\sigma$-algebra over $\mathbb{R}^n$ is denoted by $\mathcal{B}_{n}=\mathcal{B}(\mathbb{R}^{n})$, and suppose that $ \mu $ is the Lebesgue measure. 
\begin{definition}
Suppose $V\in \mathcal{B}_{n}$. A measurable function $f:V\longrightarrow\mathfrak{A}$ is called
\item[(i)] weakly integrable if $\Lambda(f)$ is Lebesgue integrable for every $\Lambda \in\mathfrak{A}^{\ast}$
 \item[(ii)] Pettis integrable if there exists  $ x\in\mathfrak{A} $ such that $ \Lambda(x)=\int_V \Lambda(f)\,\mathrm{d}\mu $, for every $\Lambda \in\mathfrak{A}^{\ast}$.
\end{definition}

If $f$ is Pettis-integrable over $V\in \mathcal{B}_{n}$ then $ x $ is unique and is called Pettis-integral of $f$ over $V$. We use the
notations $\int_V f\,\mathrm{d}\mu $ or $(P)\int_V f\,\mathrm{d}\mu$ to show the 
Pettis-integral of $f$ over $ V $. It is proved that each function $ f\in C_c(\mathbb{R}^n,\mathfrak{A}) $ is Pettis integrable over any  $V\in \mathcal{B}_{n}$ (cf. \cite{rud}).
\begin{theorem} [Change of Variables] Suppose $ D $ and $ D^{\prime} $  are open domains of
integration in $ \mathbb{R}^n $, and $ G:D\longrightarrow D^{\prime} $
 is a diffeomorphism . For every Pettis-integrable function $ f: D^{\prime} \longrightarrow \mathfrak{A} $, 
$$\int_{D^{\prime}}f\,\mathrm{d}\mu=\int_D (f\circ G)\,|\mathrm{det}(DG)|\,\mathrm{d}\mu.$$
\end{theorem}
\begin{proof}
Applying the Pettis-integral's definition and using classical change of variables theorem, one can conclude the desired result .
 \end{proof}
 \begin{definition} 
An $\mathfrak{A}$-valued $n$-form on $ M \ (n=\dim M)$ is called an $\mathfrak{A}$-valued volume form on $M$.
 \end{definition}
In the canonical coordinate system, an $\mathfrak{A}$-valued volume form $\tilde{\omega} $ on $\R^n$ is written as follows:
\begin{equation} \label{int}
\widetilde{\omega}=f\,\mathrm{d}x^1\wedge\cdots \wedge\mathrm{d}x^n,  
\end{equation}
where $f\in C^{\infty}(\mathbb{R}^n,\mathfrak{A}).$ 
\begin{definition}
Let $ \tilde{\omega} $ be a compactly supported   $\mathfrak{A}$-valued $n$-form on $\mathbb{R}^n$. Define the integral of $ \tilde{\omega} $ over $\mathbb{R}^n$ by,
\[ \int_{\mathbb{R}^n} \widetilde{\omega}=(P)\int_{\mathbb{R}^n}\,f\,\ud x_1\cdots \ud x_n \]  where   $\tilde{\omega} $ is as \eqref{int}. 
 \end{definition}

 \begin{definition} 
For $\Lambda\in\mathfrak{A}^{\ast}$, define the operator $\Lambda:\mathcal{A}^{k}(M,\mathfrak{A})\longrightarrow\mathcal{A}^{k}(M)$ by 
$$(\Lambda\tilde\omega)(X_1,...,X_k)=\Lambda(\widetilde\omega(X_1,...,X_k)),$$ where $\widetilde\omega\in\mathcal{A}^{k}(M,\mathfrak{A})$, $ X_i\in \mathfrak{X}(M). $
 \end{definition}
 \begin{definition} 
Let $M$ be an oriented smooth $n$-manifold, and let  $ \tilde{\omega} $  be a $\mathfrak{A}$-valued $n$-form on $M$. First suppose
that the support of $ \tilde{\omega} $  is compact and include in the domain of a single chart $(U, \varphi)$ which is positively oriented. We define the integral of  $ \tilde{\omega} $  over M as
\begin{equation} \label{e12}
\int_M \widetilde{\omega}=(P)\int_{\varphi(U)}(\varphi^{-1})^{\ast}(\widetilde{\omega}). \end{equation}
 \end{definition}
 
 By using the change of variable's theorem one can prove that $ \int_M \widetilde{\omega} $
 does not depend on the choice of chart whose domain contains $ \mathrm{supp} (\widetilde{\omega}). $ To integrate an arbitrary compact support $\mathfrak{A}$-valued $n$-form, we can use partition of unity as the same as ordinary $n$-forms.
 \begin{lemma} \label{stok}
 For any $\Lambda\in\mathfrak{A}^{\ast}$ we have,
 \item[(i)] if $f:M\longrightarrow N$ is a smooth map, then  for each  $\widetilde\omega\in\mathcal{A}^{n}(N,\mathfrak{A})$, $\Lambda f^*(\widetilde\omega)=f^*(\Lambda\widetilde\omega)$, 
 \item[(ii)] for a compactly supported $\mathfrak{A}$-valued volume $n$-form $\widetilde\omega $ on $ M $,
 $\Lambda(\int_M \widetilde{\omega})=\int_M
 \Lambda(\widetilde{\omega})$,
 \item[(iii)] for $\widetilde\omega\in\mathcal{A}^{k}(M,\mathfrak{A})$,  $\Lambda(\mathrm{d}\widetilde\omega)=\mathrm{d}(\Lambda\widetilde\omega)$.
 \end{lemma}
\begin{proof}
 To prove $ (i) $ suppose that $X_1,\cdots,X_k\in\XM$ , then
 \begin{align*}
f^*(\Lambda\tilde\omega)(X_1,\cdots,X_k)=&(\Lambda\tilde\omega)(f_{\ast}(X_1),\cdots,f_{\ast}(X_k))=\Lambda(\widetilde\omega(f_{\ast}(X_1),\cdots,f_{\ast}(X_k))\\=&\Lambda((f^{\ast}\widetilde\omega)(X_1,\cdots,X_k))=(\Lambda(f^{\ast}\widetilde\omega))(X_1,\cdots,X_k).
 \end{align*}
  
To prove $(ii)$ Suppose that $ \tilde{\omega} $  is compactly supported in the domain of a single chart $(U, \varphi)$ that is positively oriented, thus
\begin{align*}
\Lambda(\int_M \widetilde{\omega})=\Lambda(\int_{\mathbb{R}^n}(\varphi^{-1})^{\ast}\widetilde{\omega})=
\int_{\mathbb{R}^n}\Lambda((\varphi^{-1})^{\ast}\widetilde{\omega})
=
\int_{\mathbb{R}^n}(\varphi^{-1})^{\ast}(\Lambda\tilde{\omega})
=\int_M \Lambda\tilde{\omega}.
 \end{align*}
 
 The general case follows from the above result and using partition of unity.
 
 $(iii)$ Note that in the special case $k=0$, $\mathcal{A}^{0}(M,\mathfrak{A})=C^{\infty}(M,\mathfrak{A})$ and the continuity of $\Lambda$ implies $(iii)$ for the elements of $C^{\infty}(M,\mathfrak{A})$. The general case follows from this one. 
 \end{proof}
 \begin{theorem} [Stokes's Theorem]
 Let $ M $ be an oriented smooth $n$-manifold with boundary and orientations of $M$ and $\partial{M}$ are compatible, and let $\widetilde{\omega}$ be a compactly supported smooth valued $(n -1)$-form on $M$.Then 
 \[\int_M \mathrm{d}\widetilde{\omega}=\int_{\partial M} \widetilde{\omega}.\]
\end{theorem}
\begin{proof}
 By Lemma (\ref{stok}) for each $\Lambda\in\mathfrak{A}^{\ast}$ we have
\[
\Lambda(\int_M \mathrm{d}\widetilde{\omega})=\int_M \Lambda(\mathrm{d}\widetilde{\omega})
=\int_M\mathrm{d} (\Lambda\tilde{\omega})
=\int_{{\partial M}} (\Lambda\tilde{\omega})
=\Lambda(\int_{\partial M} \widetilde{\omega})
\]
Therefore, Han-Banach's theorem yields that
$\int_M \mathrm{d}\widetilde{\omega}=\int_{\partial M} \widetilde{\omega}$
\end{proof}
\section{\bf Connection and Curvature}
The notion of covariant derivation of $ \mathfrak{A} $-vector fields can be defined as follows.
\begin{definition}
An $ \mathfrak{A} $-connection $ \nabla $ on $ M $ is a bilinear map 
 \[
\nabla:\XA\times\XA\longrightarrow\XA
\]
such that for all $\tilde{X}, \tilde{Y}\in\XA$ and any $f\in C^{\infty}(M,\A)$, 
\begin{enumerate}
\item[$(i)$] $\nabla_{f\tilde{X}}\tilde{Y}=f\nabla_{\tilde{X}} \tilde{Y}$
\item[$(ii)$] $\nabla_{\tilde{X}}f\tilde{Y}=\tilde{X}(f)\tilde{Y}+f\nabla_{\tilde{X}}\tilde{Y}$. 
\end{enumerate}
\end{definition}
Every ordinary connection on $M$ can be extended uniquely to an $ \mathfrak{A} $-connection on $M$. For an ordinary connection on $M$ such as $\nabla$, its extension as an $\A$-connection is defined as follows. If $X,Y\in\XM$ and $h,k\in C^\infty (M,\A)$, then define
\[
\nabla_{X\otimes h} (Y\otimes k)=h((\nabla_XY)\otimes k+Y\otimes X(k))
\]

The torsion tensor of a $\A$-connection $\nabla $ is defined by
 \[
 T(\tilde{X},\tilde{Y})= \nabla_{\tilde{X}}\tilde{Y}-\nabla_{\tilde{Y}}\tilde{X}-[\tilde{X},\tilde{Y}] \qquad \tilde{X}, \tilde{Y}\in \XA 
 \]
If $ T=0 $, then $ \nabla $ is called torsion-free.
\begin{definition}
Let $ \nabla $ be an $ \mathfrak{A} $-connection on $M$. The function 
\\$\mathcal{R}:\XA\times\XA\times\XA\longrightarrow\XA$ given by 
\[ 
\mathcal{R}(\TX,\TY)(\tilde{Z})=\nabla_{\TX}\nabla_{\TY}\tilde{Z}-\nabla_{\TY}\nabla_{\TX}\tilde{Z}-\nabla_{\left[\TX,\TY\right]}\tilde{Z} 
\]
is a $ (1,3) \ \mathfrak{A} $-tensor on $ M $ and is called the curvature tensor of $ \nabla $.
\end{definition}
\begin{proposition}[The Bianchi identities] \label{bian1}
If $ \mathcal{R} $ is the curvature of a torsion free $ \mathfrak{A} $-connection $ \nabla $ on $M$, then for all $\TX, \TY, \tilde{Z}\in\XA$,
\begin{enumerate}
\item[$(i)$] $\mathcal{R}(\TX,\TY)\tilde{Z}+\mathcal{R}(\TY,\tilde{Z})\TX+\mathcal{R}(\tilde{Z},\TX)\TY=0$;
\item[$(ii)$] $(\nabla_{\TX}\mathcal{R})(\TY,\tilde{Z})+(\nabla_{\TY}\mathcal{R})(\tilde{Z},\TX)+(\nabla_{\tilde{Z}}\mathcal{R})(\TX,\TY)=0.$
\end{enumerate}
\end{proposition}
\begin{proof}
The proof is similar to ordinary connections(cf. \cite{Poor}).
\end{proof}

If $\nabla^1$, $\nabla^2$ are $ \mathfrak{A} $-connections on $M$, then, the operator $T=\nabla^1-\nabla^2$ is a $(1,2)$ $\mathfrak{A}$-tensor and the space of $\mathfrak{A}$-connections on $ M $ is an affine space which is modeled on the space of $(1,2)$ $\mathfrak{A}$-tensors.
\section{\bf Operator-Valued Inner Product }
In this section, we introduce the notion $\A$-valued inner products on $\A$-modules, and investigate some of their basic properties . For the case of Hilbert $C^\star$-modules refer to (\cite{Lan}, \cite{Lans}, \cite{Pasch}). For an  $\A$-module $V$, denote the collection of all $\mathfrak{A}$-linear mappings from $V$ into $\mathfrak{A}$ by $V^\sharp$. This space is also an $\A$-module. If $V$ is a free module, then $V^\sharp$ is also a free module.
\begin{definition}
Let $V$ be a $\A$-module. A mapping $\langle \cdot,\cdot\rangle:V\times V\longrightarrow\A$  is called an $\mathfrak{A}$-valued inner product  on $V$ if for all   $a\in\mathfrak{A}$, and $x,y,z\in V$  the following conditions hold: 
\begin{itemize}
\item[$(i)$] $\langle y,x\rangle =\langle x,y\rangle^* $
\item[$(ii)$] $\langle a x+y, z\rangle =a\langle x,z\rangle+\langle y,z\rangle$
\item[$(iii)$]  $\forall\: w\in V \langle w,x\rangle =0\Rightarrow x=0$ 
\item[$(iv)$]  $\forall\: T\in V^\sharp \:\exists\: x\in V\:\forall\: y\in V\: (T(y)=\langle y,x\rangle)$
\end{itemize}
\end{definition}
Conditions $(iii)$ and $(iv)$ are called non-degeneracy of the inner product. Note that for each $x\in V$ the mapping $\hat{x}:V\longrightarrow \mathfrak{A}$, defined by $\hat{x}(y)=\langle y,x\rangle$ belongs to $V^\sharp$. One can see that the mapping $x\longmapsto\hat{x}$ from $V$ into $V^\sharp$ is conjugate $\mathfrak{A}$-linear and non-degenracy of the inner product is equivalent to the bijectivity of this map.
 \begin{theorem} \label{ndeg}
Let $V$ be a finite dimensional free $\A$-module and $\langle \cdot,\cdot\rangle:V\times V\longrightarrow\A$ satisfies $(i)$ and $(ii)$. If $\langle \cdot,\cdot\rangle$ is non-degenerate then for any basis $\{ e_i\}$ of $V$, $\det(\langle e_i,e_j\rangle )$ is invertible in $\A$. Conversely, If for some basis $\{ e_i\}$ of $V$, $\det(\langle e_i,e_j\rangle )$ is invertible, then $\langle \cdot,\cdot\rangle$ is non-degenerate.
\end{theorem}
\begin{proof}
First, assume that for some basis $\{ e_i\}$, $\det(\langle e_i,e_j\rangle )$ is invertible in $\A$. Set $g_{ij}=\langle e_i,e_j\rangle$, so $(g_{ij})$ is an invertible $\A$-valued matrix. Denote the inversion of this matrix by $(g^{ij})$. Set $e^i=g^{ij}e_j$; $\{ e^i\}$ is also a basis and is called the reciprocal base of $\{ e_i\}$. The characteristic property of the reciprocal base is that $\langle e_i,e^j\rangle =\delta_i^j$, so for every $x\in V$, if $x=\lambda^i e_i$, then $\lambda^i=\langle x,e^i\rangle$. To prove non-degenracy, suppose for all $w\in V$ we have $\langle w,x\rangle =0$. So, for all index $i$ we have $\langle x,e^i\rangle =0$ that implies $\lambda^i=0$, so $x=0$. For any $T\in V^\sharp$, set $\lambda_i=T(e_i)$ so, for $x=\bar\lambda_ie^i$ we have $\hat x=T$.

 Conversely, suppose that $\langle \cdot,\cdot\rangle$ is non-degenerate.  Let $\{ e_i\}$ be an arbitrary basis and set $g_{ij}=\langle e_i,e_j\rangle$. For each $i$, consider the $\mathfrak{A}$-linear map $T^i:V\longrightarrow\A$ defined by $T^i(e_j)=\delta^i{_j}$. By non-degenarcy, the map $x\longmapsto\hat{x}$ is bijective and there exists $u^i\in V$ such that $\widehat{u^i}=T^i$, so $\langle e_j,u^i\rangle =\delta^i_j$. Any $u^i$ can be written as  $u^i=a^{ij}e_j$ for some $a_{ij}\in\mathfrak{A}$. An easy computation shows that the matrix $A=(a^{ij})$ is the inverse of  the matrix $ D=\left(  g_{ij}\right) $, In fact, 
\[
\delta^i_j=\langle e_j ,u^i \rangle =\langle u^i , e_j \rangle = \langle a^{ik} e_k , e_j \rangle=a^{ik}\langle e_k ,e_j\rangle =a^{ik}g_{kj}
\]
Therefore $AD=I$, so $\det(A)\det(D)=1$ and $\det(D)$ is invertible. 
\end{proof}
 Let $V$ be an $\A$ module. $V$, may have an involution that is a conjugate $\A$-linear isomorphism $*:V\longrightarrow V$ such that $*^2=1_V$. An inner product $\langle \cdot,\cdot\rangle$ on $V$ is called compatible with the involution if for all $x,y\in V$
 \[
 \langle x^*,y^*\rangle=\langle x,y\rangle^*
 \]
 In these spaces, if $x$ and $y$ are self-conjugate elements of $V$, then $\langle x,y\rangle$ is self-adjoint.
 \begin{example}
 Let $W$ be a finite dimensional real vector space and\\ $\langle \cdot,\cdot\rangle:W\times W\longrightarrow\A_{\R}$ be a symmetric bilinear map such that for some basis $\{e_i\}$ of $W$, $det(\langle e_i,e_j\rangle )$ is invertible in $\A$. The tensor product $W^\A=W\otimes_{\R}\A$ is a free $\mathfrak{A}$-module and  $\langle \cdot,\cdot\rangle$ can be extended uniquely to a $\mathfrak{A}$-valued inner product on it as follows. For all  $u,v\in V$ and $a,b\in\mathfrak{A}$ set
 \[
\langle u\otimes a\ ,v\otimes b\rangle =a\,b^{\ast}\,\langle u\ ,v\rangle .
\]
The $\A$-module $W^\A$ has a natural involution that is $(u\otimes a)^*=u\otimes a^*$ and the extended inner product on $W^\A$ is compatible with this involution. In fact, any inner product on $W^\A$ which is compatible with this involution is obtained by the above method.
 \end{example}
 
 For an inner product on a free finite dimensional $\A$-modules, we define an element of $\A$ as its signature. In the ordinary cases, signature is a scalar that is $\pm 1$ and is defined by orthonormal bases. Here, we must define signature by arbitrary bases. 
 \begin{theorem}
 Let $V$ be a free finite dimensional $\A$-module and $\langle \cdot,\cdot\rangle$ be an inner product on it. If  $\{ e_i\}$ is a basis on $V$,  set $g_{ij}=\langle e_i,e_j\rangle$ and $g={\rm det}(g_{ij})$. Then, $g$ is self-adjoint and ${{\textstyle \vert g\vert}\over{\textstyle g}}$ dose not depend on the choice of the base.
 \end{theorem}
\begin{proof}
Put $G=(g_{ij})$, since $g_{ij}^*=g_{ji}$ we find $^tG=G^*$, therefore 
\[
g={\rm det}(G)={\rm det}(^tG)={\rm det}(G^*)={\rm det}(G)^*=g^*
\]
Now, suppose that $\{e'_i\}$ is another basis. Set $g'_{ij}=\langle e'_i,e'_j\rangle$ and $G'=(g'_{ij})$ and $g'={\rm det}(G')$. For some matrix $A=(a_i^j)$ we have $e'_i=a_i^je_j$, so
\[
g'_{ij}=\langle e'_i,e'_j\rangle=\langle a_i^ke_k,a_j^le_l\rangle=a_i^ka_j^{l*}\langle e_k,e_l\rangle=a_i^ka_j^{l*}g_{kl}
\]
This equality implies that $G'=AG(^t\!A^*)$, so $g'=g\ {\rm det}(A){\rm det}(A)^*$. Since ${\rm det}(A){\rm det}(A)^*$ is a positive self-adjoint element of $\A$ we deduced that $\vert g'\vert=\vert g\vert {\rm det}(A){\rm det}(A)^*$. Consequently,
\[
{\vert g'\vert \over g'}={\vert g\vert {\rm det}(A){\rm det}(A)^* \over g\ {\rm det}(A){\rm det}(A)^*}={\vert g\vert \over g}
\]
\end{proof}
The value $\nu={{\textstyle \vert g\vert}\over{\textstyle g}}$ which does not depend on the choice of the base, is called the signature of the inner product. Note that  $\nu^{-1}=\nu$ and $\sigma(\nu)\subseteq\{-1, 1\}$. In the ordinary metrics, $\nu$ is exactly $(-1)^q$ where $q$ is the index of the inner product.

Note that in free finite dimensional $\A$ modules such as $V$ that has an $\A$ inner product, $V$ and $V^\sharp$ are naturally isomorphic and this isomorphism induces an inner product on $V^\sharp$. So, all results about $V$, can be considered for $V^\sharp$ too.
\section{\bf The Hodge Star Operator}
To define the Hodge star operator, first we need the notion of orientation of free modules. Let $V$ be a free finite dimensional $\A$-module. By definition two ordered bases of $V$ have same orientation if determinant of the transition matrix between them is a positive element of $\A$. This is an equivalence relation between the bases of $V$, but there exist many equivalence classes. We choose one of these classes as an orientation for $V$ and we call it an orientation for $V$. In this case, we call $V$ an orientated space and every basis in the orientation, is called a proper base of $V$. Note that there are many orientations on $V$ and it is not appropriate to call some of them positive.
\begin{definition}
Let $V$ be an oriented free $n$-dimensional $\A$-module that has an $\A$-inner product. For each proper base  $\{  e_i\}$ with reciprocal base $\{  e^i\}$, and $g=\mathrm{det}(g_{ij})$, set 
$\Omega=\sqrt{\left\vert g\right\vert }\,e^1\wedge\cdots\wedge e^n$. This tensor is called the canonical volume form of the inner product. 
\end{definition}
\begin{theorem}
The canonical volume form does not depend on the choice of the proper base.
\end{theorem}
\begin{proof}
Assume that $\{ u_i\}$ is another proper base with reciprocal base $\{ u^j\}$. For some matrix $A=(a_i^j)$, we have $u_i=a_i^j\,e_j$. The same orientation of $\{ e_i\}$ and $\{ u_i\}$ implies that  $\det A$ is positive. If  $A^{-1}=(\beta_i^j)$, then $e_i=\beta_i^j\,u_j$. According to the proof of Theorem (5.4)  if
\[
 g_{ij}=\langle e_i,e_j \rangle,\, u_{ij}=\langle u_i,v_j\rangle,\, g=\mathrm{det}(g_{ij}),\,  u=\mathrm{det}(u_{ij}),
 \]
 then, $u=\det A\det A^*g=(\mathrm{det}A)^2\,g$. Positivity of $\det A$  yields $\sqrt{\left\vert u\right\vert}=\mathrm{det}A.\,\sqrt{\left\vert g\right\vert}$, henceforth 
\begin{align*}
\sqrt{\left\vert g\right\vert}\;e^1\wedge\cdots\wedge e^n&=\frac{\sqrt{\left\vert g\right\vert}}{g}\,e_1\wedge\cdots\wedge e_n=
\frac{\sqrt{\left\vert g\right\vert}}{g}\;(\beta_1^{i_1}\;u_{i_1})\wedge\cdots\wedge(\beta_n^{i_n}\;u_{i_n})\\&=
\frac{\sqrt{\left\vert g\right\vert}}{g}\;\mathrm{det}A^{-1}\;u_1\wedge\cdots\wedge u_n\\&=
\frac{\sqrt{\left\vert g\right\vert}}{g\;\mathrm{det}A}\;(u_{1i_1}\;u^{i_1})\wedge\cdots\wedge(u_{ni_n}\;u^{i_n})\\&=
\frac{\sqrt{\left\vert g\right\vert}\;u}{g\;\mathrm{det}A}\;u^1\wedge\cdots\wedge u^n\\&=\frac{\sqrt{\left\vert g\right\vert}\;u\,\mathrm{det}A}{g\;(\mathrm{det}A)^2}\;u^1\wedge\cdots\wedge u^n
=\sqrt{\left\vert u\right\vert}\;u^1\wedge\cdots\wedge u^n.
\end{align*}
\end{proof}

Any $\mathfrak{A}$-inner product on the free $n$-dimensional $\A$-module $V$, in a natural way, can be extended to the space of each exterior powers $\Lambda^kV$. For 
$\alpha=u_1\wedge\cdots\wedge u_k$ and $\beta=v_1\wedge\cdots\wedge v_k$,  set $\langle\alpha,\beta\rangle=\mathrm{det}(\langle u_i,v_j\rangle)$. If $\{  e_i\}$ is a base of $V$ with reciprocal base  $\{  e^j\}$ , then it is straightforward to check that $\{ e_{i_1}\wedge\cdots\wedge e_{i_k}\}_{1\leq i_1<...<i_k\leq n}$  is a basis of $\Lambda^kV$ and its reciprocal base is $\{ e^{i_1}\wedge\cdots\wedge e^{i_k}\}_{1\leq i_1<...<i_k\leq n}$. 
\begin{remark}
Note that $\langle\Omega,\Omega\rangle=\nu$.
\end{remark}
We now define the operation $\star$, called the Hodge star operator which is similar to the Hodge star operator for ordinary metrics. This is a conjugate linear map of $\Lambda^kV$ onto $\Lambda^{n-k}V$. The operator depends on the inner product and also on the orientation of $V$. For each $\beta\in\Lambda^kV$, $\mu_\beta\in(\Lambda^{n-k}V)^{\sharp}$ is defined by
 \[
 \mu_{\beta}:\Lambda^{n-k}V\longrightarrow\mathfrak{A},\qquad \alpha\mapsto\nu\,\langle\beta\wedge\alpha,\Omega\rangle.
 \]
Because of the non-degeneracy of the inner product on $\Lambda^{n-k}V$, it follows that there is a unique $\star\beta\in\Lambda^{n-k}V$ such that $\mu_{\beta}(\alpha)=\langle\alpha,\star\beta\rangle$, that is; $\langle\alpha,\star\beta\rangle=\nu\,\langle\beta\wedge\alpha,\Omega\rangle$.
So $\star:\Lambda^kV\longrightarrow\Lambda^{n-k}V$ is an operator that for all $\beta\in\Lambda^kV$ and $\alpha\in\Lambda^{n-k}V$ we have \[
\langle\alpha,\star\beta\rangle=\nu\,\langle\beta\wedge\alpha,\Omega\rangle.
\]

The above equation shows that $\star$ is a conjugate linear map, that is i.e.
\[
\star(\beta_1+\beta_2)=\star(\beta_1)+\star(\beta_2), \qquad \star(a\beta)=a^{\ast}\star\beta.
\]
Since  $\beta\wedge\alpha$ is  a multiple of $\Omega$, and the coefficient is $\nu\,\langle\beta\wedge\alpha,\Omega\rangle$, we find that 
\[
 \forall\,\beta\in\Lambda^{k}V,\ \forall\,\alpha\in\Lambda^{n-k}V,\qquad\beta\wedge\alpha=\langle\alpha,\star\beta\rangle\Omega.
\]

We now summarize the properties of the operator $\star$ in the following theorem. 
\begin{theorem}
If $V$ is an oriented free $n$-dimensional $\A$-module that has an $\A$-valued inner product with signature $\nu$, and $\{  e_i\}$ is a proper base of $V$,  then the operator $\star$ satisfies
\begin{itemize} 
\item[(i)] $\star(e^{\sigma(1)}\wedge\cdots\wedge e^{\sigma(k)})=\mathrm{sgn}(\sigma)\,\frac{1}{\sqrt{\left\vert \textstyle{g}\right\vert}}\,e_{\sigma(k+1)}\wedge\cdots\wedge e_{\sigma(n)}$; $(\sigma\in S_n)$
\item[(ii)] $\star(e_{\sigma(1)}\wedge\cdots\wedge e_{\sigma(k)})=\mathrm{sgn}(\sigma)\,\frac{\textstyle{g}}{\sqrt{\left\vert \textstyle{g}\right\vert}}\,e^{\sigma(k+1)}\wedge\cdots\wedge e^{\sigma(n)}$; $(\sigma\in S_n)$ 
\item[(iii)] $\alpha\wedge\star\beta=\nu\,\langle\alpha,\beta\rangle\,\Omega$\quad
 $ \alpha, \beta\in\Lambda^kV$; 
\item[(iv)] $\star(1)=\nu\,\Omega $;
\item[(v)] $\star(\Omega )=1$;
\item[(vi)]$\star\star(\alpha)=(-1)^{k(n-k)}\nu\,\alpha$, \quad $\alpha\in\Lambda^kV$;
\item[(vii)]$\langle\star\alpha,\star\beta\rangle=\nu\,\langle\alpha,\beta\rangle^{\ast}$, \quad $\alpha,\beta\in\Lambda^kV.$ 
\end{itemize}
\end{theorem}
\begin{proof}
The proof is a straightforward computation and is similar to the ordinary case. 
\end{proof}
All these results, also hold for $V^\sharp$.
\section{\bf Operator-valued Metrics on Manifolds}
We now consider operator valued metrics on manifolds. 
\begin{definition}
A $\mathfrak{A}$-valued semi-Riemannian metric on an smooth manifold $M$ is a smooth map 
$\langle .,.\rangle :TM^{\A}\oplus TM^{\A} \longrightarrow \A$ such that for each $p\in M$, $\langle .,.\rangle$ restricts to a $\A$-valued inner product on $T_{p}M^{\A}$ that is compatible with its natural involution.
\end{definition}
This definition implies that inner product of any two ordinary vector fields on $M$ is an $\A_{\R}$-valued function on $M$.

For each $p\in M$, denote the signature of the inner product on $T_pM$ by $\nu_p$. The map $\nu :p\mapsto \nu_p$ is called the signature function of the metric. This function is continuous on $M$ and we can prove that it is constant on the each connected components of $M$.
\begin{theorem}
Let $M$ be a connected manifold and $\langle .,.\rangle$ be a semi-Riemannian $\A$-valued metric on it. If $\nu$ is the signature function of the metric, then $\nu$ is constant.
\end{theorem}
\begin{proof}
First, consider the following subset of $\A$.
\[
N=\{ a\in\A\ \vert\ a^*=a\ ,\ a^2=1\}
\]
Clearly, $N$ is nonempty and the values of $\nu$ are in $N$. We can prove that the distance of any two distinct elements of $N$ is greater than 2. suppose $a,b\in N$ and $a\neq b$. So, 
\[
\Vert a-b\Vert^2=\Vert(a-b)^2\Vert=\Vert a^2+b^2-2ab\Vert=2\Vert(1-ab)\Vert
\]
Since $(ab)^2=a^2b^2=1$, we have $\sigma(ab)\subset\{-1,1\}$. -1 must be in $\sigma(ab)$, otherwise $\sigma(ab)=\{ 1\}$ that implies $ab=1$, hence $a=b$ that is contrary to assumption. So, $2\in\sigma(1-ab)$ that implies $\Vert(1-ab)\Vert\geq 2$. Consequently,
\[
\Vert a-b\Vert^2=2\Vert(1-ab)\Vert\geq 4\ \Rightarrow\ \Vert a-b\Vert\geq 2
\]
This property of $N$ implies that the induced topology on $N$ is discrete, on the other hand $\nu(M)$ must be connected, so $\nu(M)$ is a singleton and $\nu$ is constant. 
\end{proof}

If $M$ is an oriented manifold, then for each $p\in M$ we can use any positive oriented basis in $T_pM$ to define an orientation for $T_pM^{\A}$. Let $\Omega_p$ be the canonical volume form on $T_{p}M^{\A}$, this gives rise to a globally defined $\A$-valued volume form $\widetilde\Omega$ over $M$. In a positive coordinate system $(U,x^1,\cdots,x^n)$, if we put
$g_{ij}=\langle \frac{\partial}{\partial x^i}, \frac{\partial}{\partial x^j} \rangle$ and
$g=\mathrm{det}(g_{ij})$, then 
\[
\widetilde{\Omega}=\sqrt{\left\vert g\right\vert }\,\mathrm{d}x^1\wedge\cdots \wedge\mathrm{d}x^n.
\]
We could also define the Hodge star operator on $\mathfrak{A}$-valued differential forms. Here, for each $p\in M$ we should consider the induced inner product on $(T_pM^{\A})^\sharp$.
\begin{definition}
The Hodge star operator
 $\star:\mathcal{A}^{k} (M,\mathfrak{A})\rightarrow \mathcal{A}^{n-k} (M,\mathfrak{A})$ maps any $k$-form $\alpha\in\mathcal{A}^{k} (M,\mathfrak{A})$ to the $(n-k)$-form $\star\alpha\in\mathcal{A}^{n-k} (M,\mathfrak{A})$ such that for any $\beta\in\mathcal{A}^{n-k} (M,\mathfrak{A})$,
\[
\alpha\wedge\beta=\langle\beta,\star\alpha\rangle\,\widetilde\Omega.
\]
\end{definition}

Our previous results now imply that 
\[
\star(1)=\nu\,\widetilde\Omega,\quad \star(\widetilde\Omega )=1,\quad\alpha\wedge\star\beta=(\beta\wedge\star\alpha)^{\ast}.
\]

For each $\alpha\in\mathcal{A}^{k} (M,\mathfrak{A})$, Also we have the identity
 \begin{align}
\label{hog} 
\star\star(\alpha)=(-1)^{k(n-k)}\,\nu\,\alpha.
\end{align}

Using the Hodge star operator, one can define a new operator, called the coderivative, denoted by
$\delta$.
\begin{definition}
The co-differential of $ \alpha\in\mathcal{A}^{k} (M,\mathfrak{A}) $ is $\delta\alpha\in\mathcal{A}^{k-1} (M,\mathfrak{A})$ defined by $$\delta\alpha=(-1)^{n(k+1)+1}\nu\,(\star\,\mathrm{d}\star\alpha).$$
\end{definition}

Note that the point-wise inner product, induces an $\mathfrak{A}$-valued inner product on $\mathcal{A}^{k}_c (M,\mathfrak{A})$ by
\[
(\alpha,\beta)=\int_M\alpha\wedge\star\beta=\int_M \nu\,\langle\alpha,\beta\rangle\,\widetilde\Omega.
\]
 The next theorem states that with respect to this inner product, the co-differential operator is adjoint to the differential operator.
\begin{theorem}
For any $\alpha\in\mathcal{A}^{k}_c (M,\mathfrak{A})$ and $\beta\in\mathcal{A}^{k-1}_c (M,\mathfrak{A})$,  
$(\mathrm{d}\beta ,\alpha)=(\beta , \delta\alpha)$
\end{theorem}
\begin{proof}
We have
 \begin{align*}
(\mathrm{d}\beta ,\alpha)=\int_M\mathrm{d}\beta\wedge\star\alpha=\int_M\mathrm{d}(\beta\wedge\star\alpha)-(-1)^{k-1}\beta\wedge\mathrm{d}\star\alpha=(-1)^k\int_M\beta\wedge\mathrm{d}\star\alpha,  
\end{align*}                                                                                                         
 On the other hand, 
\begin{align*}
(\beta , \delta\alpha)&=(\beta , (-1)^{n(k+1)+1}\nu\,(\star\,\mathrm{d}\star\alpha))\\&=\int_M\beta\wedge(-1)^{n(k+1)+1}\nu\,(\star\star\,\mathrm{d}\star\alpha)\\&=(-1)^{n(k+1)+1}\nu\,(-1)^{(n-k+1)(k-1)}\nu\,\int_M\beta\wedge\mathrm{d}\star\alpha\\&=
(-1)^k\,\int_M\beta\wedge\mathrm{d}\star\alpha.
\end{align*} 
\end{proof}

Using the operator $\delta$, we can define Laplacian operator $\Delta$ on $\mathfrak{A}$-valued differential forms, as follows
 \[
\Delta=(\mathrm{d}\delta+\delta\mathrm{d}):\mathcal{A}^{k} (M,\mathfrak{A})\longrightarrow\mathcal{A}^{k} (M,\mathfrak{A}).
\]
In particular, for a function $f$ ($0$-form) we find
 \[
 \Delta f=-\nu\,\star\,\ud\star\ud f=-\frac{1}{\sqrt{\left\vert g\right\vert }}\sum\partial_i(g^{ij}\,\sqrt{\left\vert \textstyle{g}\right\vert }\,\partial_jf).
 \]
\section {\bf The Levi-Civita Connection of Operator-valued Metrics }

Suppose that $M$ is an $\mathfrak{A}$-valued semi-Riemannian manifold, We say that a $\mathfrak{A}$-connection $\nabla$ is compatible with $\langle .,.\rangle$ if for any three vector fields $\TX, \TY, \tilde{Z}\in \XA$.
\begin{align}
\TX\langle \TY,\tilde{Z}\rangle=\langle \nabla_{\TX} \TY,\tilde{Z}\rangle +\langle \TY,\nabla_{\TX^*} \tilde{Z}\rangle
\end{align}
This equality holds iff it is hold for ordinary vector fields in $\XM$.
\begin{lemma} \label{aylin}
Suppose that $M$ is an $\mathfrak{A}$-valued semi-Riemannian manifold. For $\widetilde{V}\in \mathfrak{X}(M)^{\mathfrak{A}}$, let $\widetilde{V}^{\flat}$ be the $\mathfrak{A}$-valued one-form on $M$ given by
  \[
  \widetilde{V}^{\flat}(\TX)=\langle \widetilde{V},\TX^*\rangle, ~~~~~\mbox{for all}\quad \TX\in \XA.
  \]
Then the map $\widetilde{V}\longrightarrow\widetilde{V}^{\flat}$ is a $C^{\infty}(M,\mathfrak{A})$-module isomorphism from $\mathfrak{X}(M)^{\mathfrak{A}}$ onto $\mathcal{A}^{1}(M,\mathfrak{A})$.
\end{lemma} 
\begin{proof}
Nondegeneracy of the metric give the result.
\end{proof}
\begin{theorem} 
 If $M$ is an $\mathfrak{A}$-valued semi-Riemannian manifold, then
there exists a unique torsion free connection $\nabla$ that is compatible with the metric.
\end{theorem}
\begin{proof}
Fix $\TX, \TY\in \XA$, and define $\mu_{\TX,\TY}:\XA
\longrightarrow C^{\infty}(M,\mathfrak{A})$ by 
\[\begin{array}{rl}
\mu_{\TX,\TY}(\TZ)=&\TX\langle \TY,\TZ^*\rangle+ \TY\langle \TZ,\TX^*\rangle-\TZ\langle \TX,\TY^*\rangle \\
 &+\langle[\TX,\TY],\TZ^*\rangle-\langle[\TY,\TZ],\TX^*\rangle+\langle[\TZ,\TX],\TY^*\rangle.
\end{array}\]
A straightforward computation shows that the map $\TZ\mapsto\mu_{X,Y}(\TZ) $ is $C^{\infty}(M,\A)$-linear and is a $\mathfrak{A}$-valued one-form. By the lemma \ref{aylin}, there is an unique $\mathfrak{A}$-vector field, denoted by $2\,\nabla_{\TX}\TY$, such that $\mu_{\TX,\TY}(\TZ)=2\,\langle  \nabla_{\TX}\TY, \TZ^*\rangle$ for all $\TZ\in \XA$. Now, standard argument shows that $\nabla$ is the unique torsion free $\mathfrak{A}$-connection that is compatible with the metric. This connection is called the Levi-Civita connection of the metric. 
\end{proof}
\begin{proposition} \label{bian2}
If  $\TX,\TY,\TZ,\tilde{V}\in\XA$, then for the $\mathfrak{A}$-curvature tensor of the Levi-Civita connection, we have 
\begin{enumerate}
\item[(i)] $\langle\mathcal{R}(\TX,\TY)(\TZ),\tilde{V}\rangle=-\langle\mathcal{R}(\TX,\TY)(\tilde{V}^{\ast}),\TZ^{\ast}\rangle$;
\item[(ii)] $\langle\mathcal{R}(\TX,\TY)(\TZ),\tilde{V}\rangle=\langle\mathcal{R}(\TZ,\tilde{V}^{\ast})(\TX),\TY^{\ast}\rangle$.
\end{enumerate}
\end{proposition}
\begin{proof}
The proof is similar to the method of Semi-Riemannian manifolds (cf. \cite{one}).
\end{proof}

On $\mathfrak{A}$-valued semi-Riemannian manifolds we can straightly generalize differential operators such as gradient, and divergence.
\begin{definition}
The gradient of a function $f\in C^{\infty}(M,\mathfrak{A}) $ is $\mathfrak{A}$-vector
field that is equivalent to the differential $\ud f\in \mathcal{A}^{1}(M,\mathfrak{A})$, Thus
\[
\langle\nabla f, \TX^*\rangle=\TX(f)\qquad \forall\,\TX\in \XA
\]
in terms of a coordinate system, $\nabla f={\textstyle \partial f\over\textstyle \partial x^i}\,g^{ij} {\textstyle \partial\over\textstyle \partial x^j}$.
\end{definition}
\begin{definition}
The Hessian of a function $f\in C^{\infty}(M,\mathfrak{A}) $ is its second covariant derivative $\mathrm{Hess}(f)=\nabla(\ud f)$.
\end{definition}
The Hessian of $f$ is a symmetric $(0,2)$ $ \mathfrak{A} $-tensor field and its operation on vector fields $\TX,\TY\in\XA$ is as follows 
\[
\mathrm{Hess}(f)(\TX,\TY)=\TX(\TY f)-\langle \nabla f,\nabla_{\TX^*}\TY^* \rangle=\langle \nabla_{\TX}(\nabla f),\TY^*\rangle.
\]
\begin{definition}
If $\widetilde{X}$ is an $ \mathfrak{A} $-vector field, the contraction of its covariant differential is called divergence of $\widetilde{X}$ and is denoted by $\mathrm{div}(\widetilde{X})\in C^{\infty}(M,\mathfrak{A})$. In a  coordinate system, $\mathrm{div}(\widetilde{X})=g^{ij}\,\langle\nabla_{\partial_i}\widetilde{X},\partial_j\rangle$. 
\end{definition}
\begin{theorem}
Let $M$  be an oriented $\A$-valued semi-Riemannian manifold, and $\tilde{\Omega}$ be its canonical volume form. Then, for any $\A$-vector field $\TX\in\XA$ we have
\[
L_{\TX}\tilde{\Omega}=\mathrm{div}(\TX)\tilde{\Omega}
\]
\end{theorem}
\begin{proof}
Computations, as are done in \cite{ste} show that the equality holds in the case of scalar metrics. But, all these computations, without any change, are also valid in the case of $\A$-valued metrics, except that the functions we encountered are $\A$-valued.
\end{proof}
\section{\bf Ricci, Scalar Curvature, and Sectional Curvature}
In the past sections, we have presented the basic notions and facts about the curvature of the Levi-Civita connection of a given $\mathfrak{A}$-valued semi-Riemannian manifold. We begin to consider some  invariants that truly characterize curvature.
In this section $M$ is an $ \mathfrak{A} $-valued semi-Riemannian manifold with the $ \mathfrak{A} $-Levi-Civita connection $\nabla$.
\begin{definition}
For each $p\in M$, the Ricci curvature tensor,\\
 $ \mathcal{R}{\rm ic}_p:T_pM^{\A}\times T_pM^{\A}\longrightarrow \A $ is given by 
\[
\mathcal{R}{\rm ic}_p(\tilde{u},\tilde{v})=\mathrm{trace}\big( \ \tilde{w}\longrightarrow \mathcal{R}(\tilde{w},\tilde{u})\tilde{v}\ \big),
\]
and the scalar curvature $\mathcal{S}$ is the trace of $\mathcal{R}{\rm ic}$.
\end{definition}
In coordinate systems, 
\[
\mathcal{R}{\rm ic}(\widetilde{X},\widetilde{Y})=g^{ij}\langle\mathcal{R}(\frac{\partial}{\partial x^i},\widetilde{X})\widetilde{Y},\frac{\partial}{\partial x^j}\rangle\quad ,\quad \mathcal{S}=g^{ij}\mathcal{R}{\rm ic}(\frac{\partial}{\partial x^i},\frac{\partial}{\partial x^j})
\]

Thus, $\mathcal{R}ic$ is a symmetric $(0,2)$ tensor on $ M. $

A two-dimensional free $\mathfrak{A}$-submodule $\Pi$ of $T_pM^{\mathfrak{A}}$ is called an $\mathfrak{A}$-tangent plane to $ M $ at $p$.
For $p\in M$, $\tilde{u}, \tilde{v}\in T_pM^{\mathfrak{A}}$ define, $\mathcal{Q}(\tilde{u},\tilde{v})=\langle\tilde{u},\tilde{u}\rangle \langle\tilde{v},\tilde{v}\rangle-\langle\tilde{u},\tilde{v}\rangle\,\langle\tilde{u},\tilde{v}\rangle^{\ast}$. 
In fact, $\mathcal{Q}(\tilde{u},\tilde{v})=\langle\tilde{u}\wedge\tilde{v},\tilde{u}\wedge\tilde{v}\rangle$.
\begin{definition}
A $\mathfrak{A}$-tangent plane $\Pi$ to $ M $ is called non-degenerate if for some base $\{\tilde{u}, \tilde{v}\}$ of $\Pi$, $\mathcal{Q}(\tilde{u},\tilde{v})$ is invertible in $\mathfrak{A}$ .
\end{definition}
The invertiblility of $\mathcal{Q}(\tilde{u},\tilde{v})$ does not depend on the choice of the base. If $\{\tilde{u}, \tilde{v}\}$ is a base $ T_pM^{\mathfrak{A}}$, then $$\mathcal{K}(\tilde{u}, \tilde{v}):=\frac{\langle\mathcal{R}(\tilde{u},\tilde{v})\tilde{v}^{\ast},\tilde{u}\rangle}{\mathcal{Q}(\tilde{u},\tilde{v})}$$
is well-defined and only depends on the $2$-dimensional submodule determined
by $\tilde{u}$ and $\tilde{v}.$ 
\begin{definition}
We refer to $\mathcal{K}(\tilde{u}, \tilde{v})$ as the sectional curvature of the $2$-plane
determined by $\tilde{u}$ and $\tilde{v}.$
\end{definition}
\begin{corollary}
If $M$ has constant curvature $C$, then
\[
\mathcal{R}(\tilde{u},\tilde{v})\tilde{w}=C\{\langle\tilde{v},\tilde{w}^{\ast}\rangle\tilde{u}-\langle\tilde{u},\tilde{w}^{\ast}\rangle\tilde{v}\}.
\]
\end{corollary}




\end{document}